\documentclass{llncs}
\usepackage[utf8]{inputenc}
\usepackage{amsmath}
\usepackage{latexsym}
\usepackage{graphicx}
\usepackage{tabularx}
\usepackage{array}
\usepackage{pifont}
\usepackage{xspace}
\usepackage{todonotes}
\usepackage{paralist}
\usepackage{subfigure}
\usepackage{thm-restate}
\usepackage{lineno}

\let\doendproof\endproof
\renewcommand\endproof{~\hfill$\qed$\doendproof}

\newcommand{\simpqord}{{\sc Simultaneous PQ-Ordering}\xspace}
\newcommand{\simfpqord}{{\sc Simultaneous FPQ-Or\-de\-ring}\xspace}
\newcommand{\onefixed}{{\sc $1$-Fixed Constrained Planarity}\xspace}
\newcommand{\rowcolumn}{{\sc Row-Column Independent NodeTrix Planarity}\xspace}
\newcommand{\rcintplanarity}{{\sc RCI-NT Planarity}\xspace}
\newcommand{\rci}{{\sc RCI-NT}\xspace}
\newcommand{\nodetrix}{{\sc NodeTrix}\xspace}
\newcommand{\nodetrixplanarity}{{\sc NodeTrix Planarity}\xspace}
\newcommand{\polylink}{{\sc PolyLink}\xspace}
\newcommand{\polylinkplanarity}{{\sc PolyLink Planarity}\xspace}
\DeclareMathOperator{\fixed}{fixed}
\newcommand{\emb}{\mathcal{E}}
\newcommand{\source}{\mathrm{source}}
\newcommand{\sol}{\mathrm{sol}}

\begin{document}
	
\title{Simultaneous FPQ-Ordering and\\Hybrid Planarity Testing\thanks{Work partially supported by: MIUR, grant 20174LF3T8 AHeAD: efficient Algorithms for HArnessing networked Data; Dip. Ingegneria Univ. Perugia, grants RICBASE2017WD-RICBA18WD: ``Algoritmi e sistemi di analisi visuale di reti complesse e di grandi dimensioni''; German Science Found. (DFG), grant Ru 1903/3-1.}
}
\author{Giuseppe Liotta\inst{1}, Ignaz Rutter\inst{2}, Alessandra Tappini\inst{1}}

\date{}

\institute{Dipartimento di Ingegneria, Universit\`a degli Studi di Perugia, Italy
\email{giuseppe.liotta@unipg.it,}
\email{alessandra.tappini@studenti.unipg.it}
\and
Department of Computer Science and Mathematics, University of Passau, Germany
\email{rutter@fim.uni-passau.de}
}

\maketitle

%\linenumbers
\pagestyle{plain}

\begin{abstract}
We study the interplay between embedding constrained planarity and hybrid planarity testing. We consider a constrained planarity testing problem, called \onefixed, and prove that this problem can be solved in quadratic time for biconnected graphs. Our solution is based on a new definition of fixedness that makes it possible to simplify and extend known techniques about \simpqord. We apply these results to different variants of hybrid planarity testing, including a relaxation of \nodetrixplanarity with fixed sides, that allows rows and columns to be independently permuted.

\end{abstract}

\section{Introduction}

A \emph{flat clustered graph} $(G,S)$ consists of a graph $G$ and a set $S$ of vertex disjoint subgraphs of $G$ called \emph{clusters}. An edge connecting two vertices in different clusters is an \emph{inter-cluster edge} while an edge with both end-vertices in a same cluster is an \emph{intra-cluster edge}. A \emph{hybrid representation} of $(G,S)$ is a drawing of the graph that adopts different visualization paradigms to represent the clusters and to represent the inter-cluster edges. For example, Fig.~\ref{fi:intro-a} depicts a flat clustered graph and Fig.~\ref{fi:intro-b} shows a \nodetrix representation of this graph.

A \emph{\nodetrix representation} is a hybrid representation of a flat clustered graph where the clusters are depicted as adjacency matrices and the inter-cluster edges are drawn according to the node-link paradigm.
\nodetrix representations have been introduced to visually explore non-planar networks by Henry et al.~\cite{hfm-dhvsn-07} in one of the most cited papers of the InfoVis conference~\cite{citevis}. They have been intensively studied in the last few years, see e.g.~\cite{bbdlpp-valg-11,bdlg-ckmepd-19,ddfp-cnrcg-jgaa-17,dlpt-ntptsc-19}.

%Other types of hybrid representations have also been studied. For example, \emph{$(k,p)$ representations} relax the \nodetrix paradigm by allowing at most $p$ copies of every vertex along the boundary of a convex region~\cite{dllrt-kpprhp-walcom19}; the parameter $k$ represents the maximum size of clusters. Note that every \nodetrix representation can be regarded as a $(k,4)$ representation where $k$ is the maximum size of the matrices. Fig.~\ref{fi:intro-c} is a $(5,2)$-planar representation of the graph of Fig.~\ref{fi:intro-a}.

\polylink representations are a generalization of \nodetrix representations. In a \polylink representation every vertex of each cluster has two copies that lie on opposite sides of a convex polygon (in a \nodetrix representation the polygon is a square). For example, Fig.~\ref{fi:intro-c} shows a planar \polylink representation of the graph of Fig.~\ref{fi:intro-a}.

\emph{Intersection-link representations} are another example of hybrid representations: Each vertex of $(G,S)$ is a simple polygon and two polygons overlap if and only if there is an intra-cluster edge connecting them~\cite{addfpr-ilrg-17,aehkklnt-tcipap-gd18}. Fig.~\ref{fi:intro-d} is an intersection-link representation with unit~squares.\\
\indent Given a flat clustered graph $(G,S)$ and a hybrid representation paradigm, it makes sense to ask whether $(G,S)$ is \emph{hybrid planar}, that is, whether $(G,S)$ admits a drawing in the given paradigm such that no two inter-cluster edges cross. In general terms, hybrid planarity testing is a more challenging problem than ``traditional'' planarity testing. Hybrid representations allow more than one copy for each vertex, which facilitates the task of avoiding edge crossings but makes the problem of testing the graph for planarity combinatorially more~complex.\\
\begin{figure}[tb]
	\centering
	\subfigure[\label{fi:intro-a}{}]
	{\includegraphics[width=.33\textwidth,page=11]{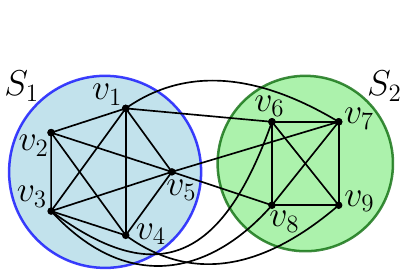}}
	\hfil
	\subfigure[\label{fi:intro-b}{}]
	{\includegraphics[width=.37\textwidth,page=12]{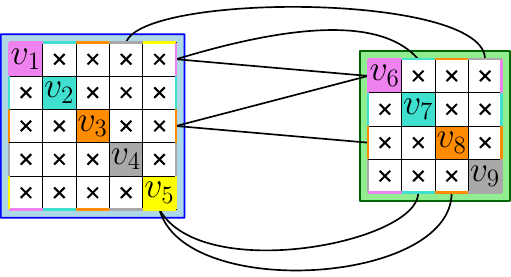}}
	\hfil 	
	\subfigure[\label{fi:intro-c}{}]
	{\includegraphics[width=.37\textwidth,page=15]{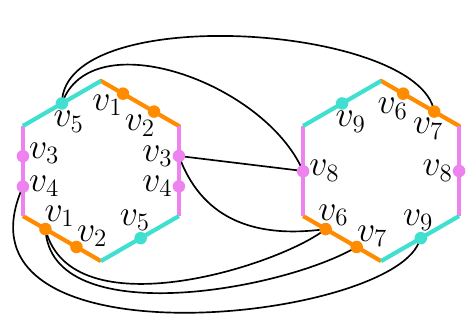}}
	\hfil 	
	\subfigure[\label{fi:intro-d}{}]
	{\includegraphics[width=.37\textwidth,page=13]{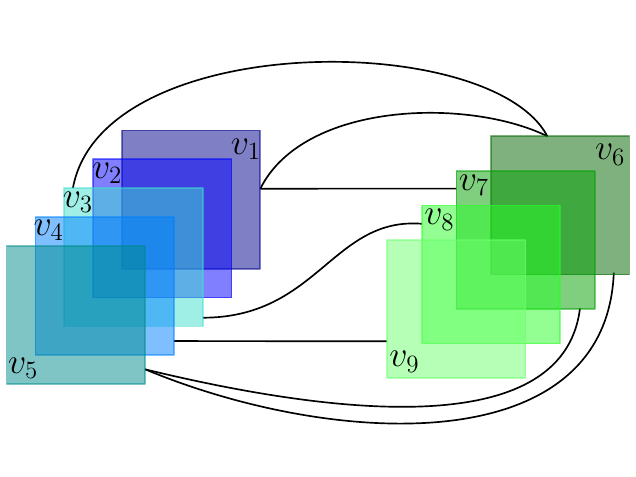}}
	\caption{(a) A flat clustered graph $(G,S)$. Clusters $S_1$ and $S_2$ are highlighted. (b) A \nodetrix representation of $(G,S)$. (c) A \polylink representation of $(G,S)$. (d) An intersection-link representation of $(G,S)$.}
	\label{fi:intro}
\end{figure}
%
%Observe that a common property of the different hybrid visualization paradigms is to allow several points of contact between vertices and inter-cluster edges. This facilitates the construction of crossing-free drawings but it makes the hybrid planarity testing problem generally more challenging than the classical graph planarity testing.
%
\indent Hybrid planarity testing can be studied both in the ``fixed sides scenario'' and in the ``free sides scenario''.
Let $e=(u,v)$ be an inter-cluster edge where $u$ is a vertex of cluster $C_u \in S$ and $v$ is a vertex of cluster $C_v \in S$. The fixed sides scenario specifies the sides of the geometric objects representing $C_u$ and $C_v$ to which edge $e$ is incident; the free sides scenario allows the algorithm to choose the sides of incidence of edge $e$. Note that what makes the problem challenging even in the fixed sides scenario is that one can permute the different copies of the vertices
along the side to which $e$ is incident.
For example, in the \nodetrixplanarity testing with fixed sides it is specified whether $e$ is incident to the top, bottom, left, or right copy of $u$ in the matrix representing $C_u$, and to the top, bottom, left, or right copy of $v$ in the matrix representing $C_v$.

%\textcolor{blue}{Fixing the sides of a $(k,p)$ representation means that it is specified to which copy of vertex $u$ and to which copy of vertex $v$ edge $e$ must be incident.}

This paper studies different variants of hybrid planarity testing in the fixed sides scenario. It adopts a unified approach that models these problems as instances of a suitably defined constrained planarity testing problem on a graph~$G$. The constrained planarity problem specifies for each vertex $v$ which cyclic orders for the edges of $G$ incident to $v$ are allowed. Choosing an order for a vertex of~$G$ influences the allowed orders for other vertices of $G$; such dependencies between different allowed orders are expressed by a directed acyclic graph
(DAG)
whose nodes are FPQ-trees (a variant of PQ-trees). Our contribution is as follows:

\begin{itemize}
	
	\item We introduce and study \onefixed and show that this problem can be solved in quadratic time for biconnected graphs, by modeling it as an instance of \simfpqord. \onefixed generalizes the partially PQ-constrained planarity testing problem studied by Bläsius and Rutter~\cite{br-spoace-16}. Our solution exploits a new definition of fixedness that simplifies and extends results of~\cite{br-spoace-16}.
	
	\item We show that a relaxation of \nodetrixplanarity with fixed sides, that allows the rows and the columns of the matrices to be independently permuted, can be modeled as an instance of \onefixed, and hence it can be solved in quadratic time if the multi-graph obtained by collapsing the clusters to single vertices is biconnected. It may be worth recalling that \nodetrixplanarity with fixed sides is NP-complete in general, but it is linear-time solvable if the rows and the columns of each matrix cannot be permuted~\cite{ddfp-cnrcg-jgaa-17}. Therefore it makes sense to further explore the conditions under which the problem is polynomially tractable.
	
	\item We introduce \polylink representations and we show that biconnected instances of \polylinkplanarity with fixed sides can be solved in quadratic time. As a byproduct, we obtain that a special instance of intersection-link planarity, called \emph{clique planarity with fixed sides}, can be solved in quadratic time. Note that clique planarity is known to be NP-complete in general~\cite{addfpr-ilrg-17}.
	We remark that \polylinkplanarity becomes equivalent to \nodetrixplanarity if the polygons have maximum size four and each side is associated with the same set of vertices.
\end{itemize}

The rest of the paper is organized as follows. Section~\ref{se:preliminaries} reports preliminary definitions. In Section~\ref{se:fixedness} we study the \onefixed testing problem and show that it can be solved in quadratic time for biconnected graphs. In Section~\ref{se:hybrid} we model the \rowcolumn testing problem as an instance of \onefixed; we also introduce the notion of \polylink representations and we show their relation with other hybrid representation paradigms.

\section{Preliminaries}\label{se:preliminaries}

\noindent \textbf{PQ-trees:} A \emph{PQ-tree} is a data structure that represents a family of permutations on a set of elements~\cite{bl-tcopiggpupa-76}. In a PQ-tree, each element is represented by a leaf node, and each non-leaf node is either a \emph{P-node} or a \emph{Q-node}. The children of a P-node can be arbitrarily permuted, while the order of the children of a Q-node is fixed up to a reversal.
Three main operations are defined on PQ-trees~\cite{br-spoace-16,bl-tcopiggpupa-76}. Let $T$ be a PQ-tree and let $L$ be the set of its leaves. Given $S \subseteq L$, the \emph{projection of $T$ to $S$}, denoted as $T|_S$,
is a PQ-tree $T'$ that represents the orders of $S$ allowed by $T$, such that $T'$ contains only the leaves of $T$ that belong to $S$.
$T'$ is obtained form $T$ by removing all leaves not in $S$ and simplifying the result, where simplifying means, that former inner nodes now having degree 1 are removed iteratively and that degree-2 nodes together with both incident edges are iteratively replaced by single edges.
The \emph{reduction of $T$ with $S$,} denoted as $T+S$,
is a PQ-tree $T'$ that represents only the orders represented by $T$ where the leaves of $S$ are consecutive.
A Q-node in $T+S$ can determine the orientation of several Q-nodes of $T$, while if we consider a P-node $\mu'$ in $T+S$, there is exactly one P-node $\mu$ in $T$ that depends on $\mu'$. We say that $\mu'$ \emph{stems from} $\mu$.
Given two PQ-trees $T_1$ and $T_2$, the \emph{intersection of $T_1$ and $T_2$},
denoted as $T_1 \cap T_2$,
is a PQ-tree $T'$ representing the orders of $L$ represented by both $T_1$ and $T_2$.
If $T_1$ and $T_2$ have the same leaves, their intersection is obtained by applying to $T_2$ a sequence of reductions with subsets of leaves whose orders are given by~$T_1$~\cite{br-spoace-16}.

\noindent \textbf{Simultaneous PQ-Ordering:} An instance of \simpqord \cite{br-spoace-16} is a DAG of PQ-trees that establishes relations between each parent node and its children nodes. Informally, the DAG imposes that the order of the leaves of a parent node must be ``in accordance with'' the order of the leaves of its children.
%The \simpqord problem asks whether there exist orders for all PQ-trees such that all child-parent relations are simultaneously satisfied.
More formally, let $N=\{T_1,\dots, T_k\}$ be a set of PQ-trees whose leaves are $L(T_1), \dots, L(T_k)$, respectively.
Let $\mathcal{I} = (N,Z)$ be a DAG with vertex set $N$ and such that every arc in $Z$ is a triple $(T_i, T_j; \varphi)$ where $T_i$ is the tail vertex, $T_j$ is the head vertex, and $\varphi : L(T_j) \rightarrow L(T_i)$ is an injective mapping from the leaves of $T_j$ to the leaves of $T_i$ ($1 \le i,j \le k$). Given two cyclic orders $O_i$ and $O_j$ defined by $T_i$ and $T_j$, respectively, we say that $O_i$ \emph{extends} $\varphi(O_j)$ if $\varphi(O_j)$ is a suborder of $O_i$.
The \simpqord problem asks whether there exist cyclic orders $O_1, \dots O_k$ of $L(T_1), \dots, L(T_k)$, respectively, such that for each arc $(T_i, T_j; \varphi) \in Z$, $O_i$ extends $\varphi(O_j)$.
Let $(T_i, T_j; \varphi)$ be an arc in $Z$. An internal node $\mu_i$ of $T_i$ is \emph{fixed by} an internal node $\mu_j$ of $T_j$ (and $\mu_j$ \emph{fixes} $\mu_i$ in $T_i$) if there exist leaves $x, y, z \in L(T_j)$ and $\varphi(x), \varphi(y), \varphi(z) \in L(T_i)$ such that (i) removing $\mu_j$ from $T_j$ makes $x$, $y$, and $z$ pairwise disconnected in $T_j$, and (ii) removing $\mu_i$ from $T_i$ makes $\varphi(x)$, $\varphi(y)$, and $\varphi(z)$ pairwise disconnected in $T_i$.\\
%Note that by (i) the three paths connecting $\mu_j$ with $x$, $y$, and $z$ in $T_j$ share no node other than $\mu_j$, while by (ii) those connecting $\mu_i$ with $\varphi(x)$, $\varphi(y)$, and $\varphi(z)$ in $T_i$ share no node other than $\mu_i$. Since any order $O_j$ determines a cyclic order around $\mu_j$ of the paths connecting it with $x$, $y$, and $z$ in $T_j$, any order $O_i$ extending $\varphi(O_j)$ determines the same cyclic order around $\mu_i$ of the paths connecting it with $\varphi(x)$, $\varphi(y)$, and $\varphi(z)$ in $T_i$; this is why we say that $\mu_i$ is fixed by $\mu_j$.
%
\indent An instance $\mathcal{I} = (N,Z)$ of \simpqord is \emph{normalized} if, for each arc $(T_i, T_j; \varphi) \in Z $ and for each internal node $\mu_j \in T_j$, tree $T_i$ contains exactly one node $\mu_i$ that is fixed by $\mu_j$.
Every instance of \simpqord can be normalized by means of an operation called the \emph{normalization}~\cite{addfpr-blp-16,br-spoace-16}, which is defined as follows. Consider each arc $(T_i, T_j;\varphi) \in Z$ and replace $T_j$ with $T_i|_{\varphi(L(T_j))} \cap T_j$ in $\mathcal{I}$, that is, replace tree $T_j$ with its intersection with the projection of its parent $T_i$ to the set of leaves of $T_i$ obtained by applying mapping $\varphi$ to the leaves $L(T_j)$ of $T_j$.
Consider a normalized instance $\mathcal{I} = (N,Z)$. Let $\mu$ be a P-node of a PQ-tree $T$ with parents $T_1,\dots, T_p$ and let $\mu_i \in T_i$ be the unique node in $T_i$, with $1 \le i \le p$, fixed by $\mu$. The \emph{fixedness} of $\mu$ is defined as $\fixed(\mu) = \omega + \sum_{i=1}^{p}(\fixed(\mu_i)- 1)$, where $\omega$ is the number of children of $T$
containing a node that fixes $\mu$.
%fixing $\mu$.
A P-node $\mu$ is \emph{$k$-fixed} if $\fixed(\mu)\le k$. Also, instance $\mathcal{I}$ is \emph{$k$-fixed} if all the P-nodes of any PQ-tree $T \in N$ are $k$-fixed.

\noindent \textbf{FPQ-trees:} An \emph{FPQ-tree} is a PQ-tree where, for some of the Q-nodes, the reversal of the permutation described by their children is not allowed. To distinguish these Q-nodes from the regular Q-nodes, we call them \emph{F-nodes}~\cite{lrt-gpthec-19}.
%The three operations defined for PQ-trees can be immediately extended to the case of FPQ-trees.
%An \emph{FPQ-constraint} for a vertex $v$ consists of associating with $v$ one or more FPQ-trees that specify a set of allowed cyclic orders of the edges incident to $v$.
%The \simfpqord problem is a variant of \simpqord, where the nodes of the input DAG are FPQ-trees.
The study of Bläsius and Rutter on \simpqord also considers the case in which the permutations described by some of the Q-nodes are totally fixed, hence the results given in~\cite{br-spoace-16} for \simpqord also hold when the nodes of the input DAG are FPQ-trees. In the rest of the paper we talk about \simfpqord to emphasize the presence of F-nodes, since they play an important role in our applications of hybrid planarity testing.

\noindent \textbf{Embedding DAG:}
Let $G$ be a biconnected planar graph and let $\mathcal{T}$ be an SPQR-decomposition tree of $G$.
The \emph{embedding DAG of $G$}, denoted as $\mathcal{D}$, describes for each vertex $v\in V$, the cyclic orders in which its incident edges appear in any planar embedding of $G$.
These cyclic orders can be described by looking at the SPQR-decomposition tree of $G$. We can ``translate'' an SPQR-decomposition tree $\mathcal{T}$ of $G$ into a set of PQ-trees (the embedding trees), which represent the same combinatorial embeddings as the ones defined by $\mathcal{T}$~\cite[Section~$2.5$]{br-spoace-16}. Note that the cyclic orders around a vertex depend in general on the cyclic orders of the edges around other vertices.

Bl\"asius and Rutter describe how to express such dependencies and all the planar embeddings of a graph into a DAG of PQ-trees by describing and exploiting the relation between PQ-trees and SPQR-trees~\cite[Section~$4.1$]{br-spoace-16}. The obtained DAG of PQ-trees is the embedding DAG $\mathcal{D}$, whose size is linear in the size of the SPQR-tree $\mathcal{T}$, and thus linear in the size of $G$ itself. $\mathcal{D}$ has an embedding tree $T(v)$ for each vertex $v$ of $G$, and other PQ-trees are connected to $T(v)$ in order to encode the dependencies with the cyclic orders of other vertices.
Figure~\ref{fi:example2-b} shows an SPQR-decomposition tree of the graph $G$ in Figure~\ref{fi:example2-a}, while
Figure~\ref{fi:example2-c} shows the embedding DAG $\mathcal{D}$ of the graph $G$ in Figure~\ref{fi:example2-a}, which encodes all the embedding constraints for the graph $G$.
Note that $\mathcal{D}$ has only P- and Q-nodes.

%\begin{figure}[htb]
%	\centering
%	\subfigure[\label{fi:example-a}{}]
%	{\includegraphics[width=.22\textwidth,page=1]{exx}}
%	\hfill
%	%	\subfigure[\label{fi:example-b}{}]
%	%	{\includegraphics[width=.8\textwidth,page=2]{exx}}
%	%	\hfil
%	\subfigure[\label{fi:example-c}{}]
%	{\includegraphics[width=.77\textwidth,page=3]{exx}}
%	\caption{(a) A multi-graph $G$ where a parallel component is highlighted. 
%		%		(b) An SPQR-decomposition tree of $G$.
%		(b) The embedding DAG $\mathcal{D}$ of $G$.}
%	\label{fi:running-example}
%\end{figure}

\begin{figure}[tb]
	\centering
	\subfigure[\label{fi:example2-a}{}]
	{\includegraphics[width=.2\textwidth,page=1]{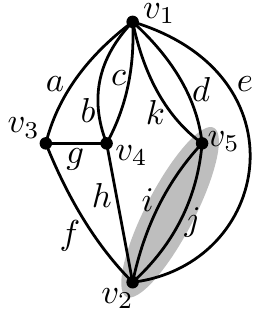}}
	\hfill
	\subfigure[\label{fi:example2-b}{}]
	{\includegraphics[width=.79\textwidth,page=2]{ex}}
	\hfill
	\subfigure[\label{fi:example2-c}{}]
	{\includegraphics[width=.9\textwidth,page=3]{ex}}
	\caption{(a) A biconnected planar graph $G$.
		(b) An SPQR-decomposition tree of $G$.
		(c) The embedding DAG $\mathcal{D}$ of $G$. P-nodes are depicted as circles, while Q-nodes are boxes.}
	\label{fi:running-example2}
\end{figure}

For example, $T(v_1)$ in Figure~\ref{fi:example2-c} is a PQ-tree that describes all the possible cyclic orders that the edges incident to $v_1$ can have in a planar embedding of the graph $G$ in Figure~\ref{fi:example2-a}: Edges $b$ and $c$ can be arbitrarily permuted, as well as edges $d$ and $k$, while $a$ cannot be found between $b$ and $c$.
The node $\delta^{\sf I}$ in $T(v_2)$ and the node $\delta^{\sf II}$ in $T(v_5)$ are two \emph{consistent} P-nodes, which means that in any planar embedding of $G$, if edge $i$ is encountered after edge $j$ in counter-clockwise order around $v_2$, then edge $i$ must be encountered before edge $j$ in counter-clockwise order around $v_5$. This constraint depends on the fact that $v_2$ and $v_5$ are the poles of a same triconnected component of $G$, highlighted in gray in Figure~\ref{fi:example2-a}.
This constraint is described by the PQ-tree $P(\delta)$ and by the two edges that are directed from $T(v_2)$ to $P(\delta)$ and from $T(v_5)$ to $P(\delta)$; one of these edges is labeled as \emph{reversing} because the orders of the edges around the two vertices must be opposite to one another. The injective mapping between source PQ-trees and sink PQ-trees of $\mathcal{D}$ is not shown in Figure~\ref{fi:example2-c}, but the starting points of the arcs suggest which mappings are suitable. For example, a suitable mapping is between the three leaves $b$, $c$, and $g$ of $T(v_4)$ and the three leaves of $P(\alpha)$; while a suitable mapping between $T(v_4)$ and $Q(\beta)$ maps $g$, $h$ and $b$ to the leaves of $Q(\beta)$.

\section{Fixedness and $1$-Fixed Constrained Planarity}\label{se:fixedness}

Bläsius and Rutter in~\cite{br-spoace-16} show that normalized instances of \simfpqord can be solved in quadratic time if they are $2$-fixed. In their applications, all instances are already normalized (or have a very simple structure) so that it is easy to verify whether an instance is $2$-fixed. The difficulty of applying their result to other contexts is that if the instances are not normalized, it is quite technical to understand the structure of the normalized instance and to check whether it is $2$-fixed.
In this section we present a new definition of fixedness that does no longer require the normalization as a preliminary step to check whether an instance of \simfpqord is $2$-fixed. This definition, given in Section~\ref{sse:def} together with the notion of joinable instances, significantly simplifies the application of \simfpqord. In Section~\ref{sse:constr-pl} we discuss the impact of this definition to efficiently solve a constrained planarity testing problem, called \onefixed.

\subsection{A New Definition of Fixedness}\label{sse:def}

\begin{definition}\label{de:fixedness}
  Let $\mathcal{I}=(N,Z)$ be an instance of \simfpqord and let $\mu$
  be a P-node of an FPQ-tree that belongs to a node $v$ of
  $\mathcal{I}$.
  The {\em fixedness} of $\mu$ is denoted as $\fixed(\mu)$.
  %	If $v$ is a source,
  Let $\omega$ be the number of children of $v$ fixing $\mu$.  If $v$
  is a source, we define $\fixed(\mu) = \omega$. If $v$ is not a
  source, let~$p$ be the number of parent nodes $T_1,\dots,T_{p}$
  of~$v$ in $\mathcal{I}$. For~$i=1,\dots,p$, let~$F_i$ be the set
  of P-nodes of $T_i$ that is fixed by~$\mu$. If $|F_i| = 0$ for some
  $i=1,\dots,k$, then $\fixed(\mu) = 0$, otherwise~$\fixed(\mu) = \omega+\sum_{i=1}^{p} \max_{\nu \in
    F_i}(\fixed(\nu)-1)$. The P-node $\mu$ is \emph{$k$-fixed} if
  $\fixed(\mu)\le k$.  Instance $\mathcal{I}$ is \emph{$k$-fixed} if all  P-nodes of FPQ-trees $T \in N$ are~$k$-fixed.
\end{definition}

We remark that Definition~\ref{de:fixedness} coincides with the notion of fixedness given in~\cite{br-spoace-16} (see Section~\ref{se:preliminaries}) if we restrict ourselves to normalized instances. Namely, in a normalized instance, $|F_i| = 1$ for $i=1,\dots,p$, and the maximum vanishes.

\begin{restatable}{lemma}{fixednessnotincreasing}\label{le:fixedness-notincreasing}
%	\textsc{[*]}
	Let $\mathcal{I}$ be an instance of \simfpqord and let~$\mathcal{I}'$ be the
	normalization of $\mathcal{I}$.  Then $\fixed(\mathcal{I}') \le \fixed(\mathcal{I})$.
\end{restatable}
\begin{proof}
	We recall that normalizing an arc $(T,T';\varphi)$ of an instance of
	\simfpqord means replacing $T'$ by
	$T''=T|_{\varphi(L(T'))} \cap T'$, i.e., we first project $T$ to the
	leaves of $T'$, which yields a tree $T^\star$ whose leaves
	bijectively correspond to those of $T'$ and then intersect $T^\star$
	and $T'$ to obtain~$T''$. In this way, any restrictions that $T$
	imposes on the ordering of the leaves of $T'$ are transferred to
	$T''$, which thus represents exactly those orders of $T'$ that can
	be extended to $T$. The normalization process executes this in
	top-down fashion for each arc of the instance, thus giving a
	sequence of instances
	$\mathcal{I}=\mathcal{I}_0,
	\mathcal{I}_1,\mathcal{I}_2,\dots,\mathcal{I}_m = \mathcal{I}'$,
	where $m$ is the number of arcs of
	$\mathcal{I}$~\cite{br-spoace-16}.  We prove that
	$\fixed(\mathcal{I}_{i+1}) \le \fixed(\mathcal{I}_{i})$ for
	$i=0,\dots,m-1$, which implies the claim of the lemma.
	
	Assume that $\mathcal{I}_{i+1}$ is obtained from $\mathcal{I}_i$ by
	normalizing an arc $(T,T';\varphi)$ to $(T,T'';\varphi)$.
	Let~$\mu''$ be a P-node of $T''$.  Since $T''$ is obtained by
	applying to $T^\star$ a sequence of reductions with subsets of
	leaves whose orders are given by $T'$, we have that $\mu''$ stems
	from a single P-node~$\mu'$ of $T'$.  We have the following claim.
	
	\begin{claim}
		$\fixed(\mu'') \le \fixed(\mu')$
	\end{claim}
	
	We first show that by using this claim, the inductive step of the
	lemma follows and then prove the claim.  The fixedness of a
	node~$\mu$ depends on the number of children fixing it, as well as,
	for each parent in which it fixes a P-node the maximum fixedness of
	those P-nodes.  First observe that, whether a P-node~$\mu$ of some
	arbitrary FPQ-tree $T_1$ is fixed by one of its children $T_2$ or not
	depends solely on the set of leaves $L(T_2)$ and not on any other
	structural considerations.  Since the normalization of an arc does
	not change the leaf set of any tree, for each P-node not in $T''$
	the number of children fixing it does not change.  For a
	P-node~$\mu''$ of $T''$, any child that fixes~$\mu''$ also fixes the
	node~$\mu'$ it stems from, and therefore the number $w''$ of
	children fixing~$\mu''$ is upper-bounded by the number $w'$ of
	children fixing~$\mu'$.  Therefore, if a P-node~$\nu$ not in $T''$
	increases its fixedness, this is due to a parent FPQ-tree containing
	a P-node~$\nu'$ that is fixed by~$\nu$ for which $\fixed(\nu')$
	increased by the normalization.  Traversing the DAG upwards, in this
	way, we eventually find a P-node~$\mu''$ of~$T''$ that is
	responsible for the increase in fixedness.  But then, before the
	normalization, the fixedness of the node~$\mu'$ from which $\mu''$
	stems was used to compute the fixedness of the corresponding child.
	The claim implies that this fixedness is greater than or equal to
	the new value that is used to determine the fixedness.  This
	contradicts the assumption that the fixedness of~$\nu$ increased.
	
	It remains to prove the claim.  As argued above, the
	numbers~$\omega'$ and~$\omega''$ of children of $T'$ and~$T''$ that
	fix~$\mu'$ and~$\mu''$, respectively,
	satisfy~$\omega'' \le \omega'$.  Similarly, let~$p'$ and~$p''$ be
	the number of parent nodes for which a P-node is fixed by~$T'$
	and~$T''$, respectively.  In particular, $T'$ and~$T''$ fix P-nodes
	from the same parent trees~$T_1,\dots,T_{p'}$.  Let~$F_i'$
	and~$F_i''$ denote the sets of P-nodes of $T_i$ that are fixed
	by~$\mu'$ and~$\mu''$, respectively.  Again, since~$\mu''$ stems
	from~$\mu'$, it follows that~$F_i'' \subseteq F_i'$
	for~$i=1,\dots,p'$.  Moreover, the fixedness of the nodes in the
	sets $F_i$ did not increase, since they are not descendants of
	$T''$.  Therefore the claim follows.
\end{proof}

By Lemma~\ref{le:fixedness-notincreasing}, it suffices to check the $2$-fixedness of a non-normalized instance of \simfpqord to conclude that it can be solved in quadratic time by exploiting~\cite[Theorems 3.11, 3.16]{br-spoace-16}. We now further simplify the applicability of the result.

Let $\mathcal{I} = (N,A)$ be an instance of \simfpqord.  We denote
by~$\source(\mathcal{I})$ the set of sources of $\mathcal{I}$.  A solution of an
instance~$\mathcal{I}=(N,A)$ of \simfpqord determines a tuple of cyclic orders
$(O_v)_{v \in N}$.  In many cases, we are only interested in the
cyclic orders at the sources, and we therefore define
$\sol(\mathcal{I}) = \{ (O_v)_{v \in \source(\mathcal{I})} \mid \mathcal{I}$ has a solution
$(O_v')_{v \in N}$ with $O_v = O_v'$ for $v \in \source(\mathcal{I})\}$.  We say
that an instance $\mathcal{I}$ has \emph{P-degree} $k$ if every node whose
FPQ-tree contains a P-node has at most $k$ parents.
Let~$\mathcal{I}$ and~$\mathcal{I}'$ be two instances of \simfpqord such that there exists a
bijective mapping~$M$ between the sources of~$\mathcal{I}$ and the sources
of~$\mathcal{I}'$ with $L(M(T)) = L(T)$ for each source $T$ of $\mathcal{I}$.  We call $\mathcal{I}$
and~$\mathcal{I}'$ \emph{joinable}. 
%The \emph{join DAG} of~$\mathcal{I}$ and~$\mathcal{I}'$ is the
%instance $\mathcal{I} \Join \mathcal{I}'$ obtained by replacing, for each source $T$ of $\mathcal{I}$, the
%nodes $T$ and~$M(T)$ by $T \cap M(T)$ and identifying the respective
%nodes of $\mathcal{I}$ and~$\mathcal{I}'$.  By construction, it is
%$\sol(\mathcal{I} \Join \mathcal{I}') = \sol(\mathcal{I}) \cap \sol(\mathcal{I}')$.
The \emph{join DAG} of~$\mathcal{I}$ and~$\mathcal{I}'$ is the instance $\mathcal{I} \Join \mathcal{I}'$ obtained by replacing, for each source node $T$ of $\mathcal{I}$ (and each corresponding source node $M(T)$ of $\mathcal{I}'$), the nodes $T$ (and~$M(T)$) by $T \cap M(T)$ and identifying the respective nodes of $\mathcal{I}$ and~$\mathcal{I}'$.  By construction, it is $\sol(\mathcal{I} \Join \mathcal{I}') = \sol(\mathcal{I}) \cap \sol(\mathcal{I}')$.

\begin{restatable}{lemma}{lemmaIntersectiontwofixed}\label{le:intersection-2fixed}
%	\textsc{[*]}
  Let $\mathcal{I}$ and $\mathcal{I}'$ be joinable instances of \simfpqord with P-degree at
  most $2$ and such that their associated DAGs each have height~$1$.  If
  both~$\mathcal{I}$ and~$\mathcal{I}'$ are $1$-fixed, then $\mathcal{J} = \mathcal{I} \Join  \mathcal{I}'$ is $2$-fixed.
\end{restatable}
\begin{proof}
	By construction, the height of~$\mathcal{J}$ is~$1$, i.e., each node is either a
	source or a sink.  We show that the fixedness of each P-node of an
	FPQ-tree of $\mathcal{J}$ is at most~$2$.  We treat the sources and sinks
	separately.  Let~$\mu$ be a P-node of a source $T$ of $\mathcal{J}$.
	Since~$T$ is the intersection of a source $S$ of $\mathcal{I}$ with a source
	$S'$ of $\mathcal{I}'$, it follows that~$\mu$ stems from a single P-node~$\nu$
	of $S$ and from a single P-node~$\nu'$ of $S'$.  Clearly, any child of $\mathcal{I}$
	that fixes~$\mu$ must either have fixed~$\nu$ or~$\nu'$.  Hence,
	$\fixed(\mu) \le \fixed(\nu) + \fixed(\nu') \le 2$ since~$\mathcal{I}$
	and~$\mathcal{I}'$ are $1$-fixed.
	
	Let now $\mu$ be a P-node of a sink $T$ of $\mathcal{J}$ that has at least one
	parent (otherwise $T$ would be a source).  Due to the above, P-nodes
	of all sources are at most $2$-fixed.  Hence~$\fixed(\nu) - 1 \le 1$
	for each P-node~$\nu$ of a parent that is fixed by~$\mu$.  Since $\mathcal{I}$
	and~$\mathcal{I}'$ have P-degree at most~$1$, $T$ has at most two parents.  It hence follows that $\fixed(\mu) \le 2$.
\end{proof}

\subsection{1-Fixed Constrained Planarity}\label{sse:constr-pl}

Let $G=(V,E)$ be a biconnected planar graph, let~$v \in V$ be a
vertex, and let $E(v)$ be the edges of $G$ incident to $v$. A \emph{$1$-fixed constraint} $\mathcal{C}(v)$ for $v$ is a $1$-fixed instance of \simfpqord such that it has P-degree at most~$2$ and it has a single source whose FPQ-tree has the edges in $E(v)$ as its leaves.
The following property is implied by~\cite[Section~4.1]{br-spoace-16}.

\begin{property}\label{pr:embeddingDAGlength1}
	For each vertex $v$ of $G$, $\mathcal{D}(v)$ is a $1$-fixed constraint.
\end{property}

Let $\emb$ be an embedding of $G$ and let $\emb(v)$ be the cyclic order that $\emb$ induces on the edges around $v$.
We say that embedding~$\emb$ \emph{satisfies} constraint $\mathcal{C}(v)$ if there exists a solution for~$\mathcal{C}(v)$ such that the order of the source is~$\emb(v)$.

Given a graph $G$ and a $1$-fixed constraint for each vertex of $G$, the \onefixed testing problem asks whether $G$ is \emph{$1$-fixed constrained planar}, i.e., it admits a planar embedding that satisfies all the constraints.

\begin{restatable}{theorem}{theoremConstrainedEmbeddingQuadratic}\label{th:constrained-embedding-quadratic}
%  Let $G = (V,E)$ be a biconnected planar graph with $n$ vertices, and
%  for each~$v \in V$ let $\mathcal{C}(v)$ be a $1$-fixed
%  constraint.  Then it can be tested in $O(n^2)$ time whether~$G$
%  admits a planar embedding~$\emb$ that satisfies all constraints.
  Let $G = (V,E)$ be a biconnected planar graph with $n$ vertices, and
  for each~$v \in V$ let $\mathcal{C}(v)$ be a $1$-fixed
  constraint. \onefixed can be tested in $O(n^2)$ time.
\end{restatable}
\begin{proof}
  Let $\mathcal{D}$ be the embedding DAG of~$G$, where $\sol(\mathcal{D})$
  corresponds bijectively to the rotation systems of the planar embeddings
  of $G$~\cite{br-spoace-16}.
  The embedding DAG $\mathcal{D}(v)$ of a vertex $v \in V$ is such that  $\sol(\mathcal{D}(v))$ corresponds bijectively to the cyclic orders that the planar embeddings of $G$ induce around $v$.
  Let $\mathcal{C}$ denote the instance of
  \simfpqord that is the disjoint union
  $\bigcup_{v \in V} \mathcal{C}(v)$, and observe further that~$\sol(\mathcal{C})$ are
  precisely the rotations at vertices that satisfy all the constraints
  $\mathcal{C}(v)$.  Observe further that $\mathcal{D}$ and~$\mathcal{C}$ are joinable, and
  $\sol(\mathcal{D} \Join \mathcal{C})$ are exactly the rotation systems of planar
  embeddings of $G$ that satisfy all the constraints $\mathcal{C}(v)$, $v \in V$.
By Property~\ref{pr:embeddingDAGlength1}, both $\mathcal{D}$ and~$\mathcal{C}$ are $1$-fixed, have height~$1$ and P-degree at most~$2$.
  Therefore, by Lemma~\ref{le:intersection-2fixed} $\mathcal{J} = \mathcal{D} \Join \mathcal{C}$
  is $2$-fixed and by Lemma~\ref{le:fixedness-notincreasing} also the
  normalization of $\mathcal{J}$ is $2$-fixed.  It follows that the normalization
  of $\mathcal{J}$ can be solved in $O(n^2)$ time~\cite[Theorems 3.11,
  3.16]{br-spoace-16}.  The overall result follows from the fact that
  $\mathcal{D}$ and~$\mathcal{C}$ have size linear in $n$ and their normalization can be
  computed in linear time~\cite[Lemma 3.12]{br-spoace-16}.
\end{proof}

\section{Hybrid Planarity Testing Problems}\label{se:hybrid}

% JOURNAL VERSION
%A \emph{flat clustered graph} $(G,S)$ consists of a graph $G$ and a set $S$ of vertex disjoint subraphs of $G$ called \emph{clusters}.
%In a \emph{hybrid representation} of a flat clustered graph, the classical node-link paradigm is used to represent the global structure of the graph, while the clusters are represented according to a different representation paradigm inside disjoint geometric regions. The edges that connect vertices belonging to distinct clusters are the inter-cluster edges, while the edges connecting two vertices of a same cluster are the intra-cluster edges. Different hybrid representations of a same graph are possible by defining different rules to represent the clusters.
%Given a hybrid representation paradigm and a flat clustered graph $(G,S)$, the \emph{hybrid planarity testing} problem asks whether $G$ admits a hybrid representation such that inter-cluster edges do not cross each other and they do not cross any cluster.

In this section, we study the interplay between hybrid planarity testing problems and \onefixed. We consider a variant of the \nodetrix paradigm in Section~\ref{sse:rcint}, and we study \polylink representations in Section~\ref{sse:other-hybrid}, which include special cases of intersection-link representations~\cite{addfpr-ilrg-17}.
%and $(k,p)$ representations~\cite{dllrt-kpprhp-walcom19}.

\subsection{The Row-Column Independent NodeTrix Planarity Problem}\label{sse:rcint}

We recall that in a \nodetrix representation each cluster is represented as an adjacency matrix, while the inter-cluster edges are simple curves connecting the corresponding matrices and not crossing any other matrix~\cite{ddfp-cnrcg-jgaa-17,dlpt-ntptsc-19,hfm-dhvsn-07}.
%A \emph{\nodetrix graph} is a flat clustered graph that admits a \nodetrix representation.
A \emph{\nodetrix graph} is a flat clustered graph with a \nodetrix representation.
For example, Fig.~\ref{fi:intro-b} is a \nodetrix representation of the graph in Fig.~\ref{fi:intro-a}; note that in this representation for every vertex there are four segments, one for each side of the matrix, to which inter-cluster edges can be connected. A \nodetrix representation is said to be \emph{with fixed sides} if the sides of the matrices to which the inter-cluster edges must be incident are given as part of the input.
%\nodetrixplanarity with fixed sides is the problem of deciding whether a flat clustered graph $(G,S)$ admits a \nodetrix representation with fixed sides where no pair of inter-cluster edges cross.

%In this paper we consider the \nodetrixplanarity testing problem in the scenario with fixed sides.

%Da~Lozzo et al.~\cite{ddfp-cnrcg-jgaa-17} showed that \nodetrixplanarity testing with fixed sides is NP-hard, and Di~Giacomo et al.~\cite{dlpt-ntptsc-19} showed that the problem is fixed parameter tractable with respect to the maximum size of the clusters if the graph obtained by collapsing each cluster in $S$ into a single vertex has tree-width at most two.

The \nodetrixplanarity testing problem with fixed sides is NP-hard~\cite{ddfp-cnrcg-jgaa-17}, and it is fixed parameter tractable with respect to the maximum size of clusters and to the branchwidth of the graph obtained by collapsing each cluster into a single vertex, as shown in~\cite{dlpt-ntptsc-19,lrt-gpthec-19}. \nodetrixplanarity with fixed sides is known to be solvable in linear time when rows and columns are not allowed to be permuted~\cite{ddfp-cnrcg-jgaa-17}.
%(it suffices to replace each cluster with a suitable gadget and to perform a planarity testing~\cite{ht-ept-74}).
This naturally raises the question about whether a polynomial-time solution exists also for less constrained versions of \nodetrixplanarity.

We study the scenario in which the permutations of rows and columns can be chosen independently. Namely, we introduce a relaxed version of \nodetrixplanarity with fixed sides, called \rowcolumn (\rcintplanarity for short). \rcintplanarity asks whether a flat clustered graph admits a planar \nodetrix representation in the fixed sides scenario, but it allows to permute the rows and the columns independently of one another. A graph for which the \rcintplanarity test is positive is said to be \emph{\rci planar}.

\noindent \textbf{{The Equipped Frame Graph:}} We model \rcintplanarity as an instance of \onefixed defined on a (multi-)graph associated with $(G,S)$, that we call the \emph{equipped frame graph} of $G$, denoted as $G_F$.
Graph $G_F$ is obtained from $G$ by collapsing each cluster into a single vertex. More precisely, $G_F$ has $n_F=|S|$ vertices, each one corresponding to one of the matrices defined by $S$. There is an edge between two vertices $u$ and $v$ of $G_F$ if and only if there is an edge in $G$ between matrices $M_u$ and $M_v$ corresponding to $u$ and to $v$, respectively.
A \nodetrix graph is biconnected if its equipped frame graph is biconnected and, from now on, we consider biconnected \nodetrix graphs.

\begin{figure}[tb]
	\centering
	\subfigure[\label{fi:matrix-simpq-a}{}]
	{\includegraphics[width=.48\textwidth,page=2]{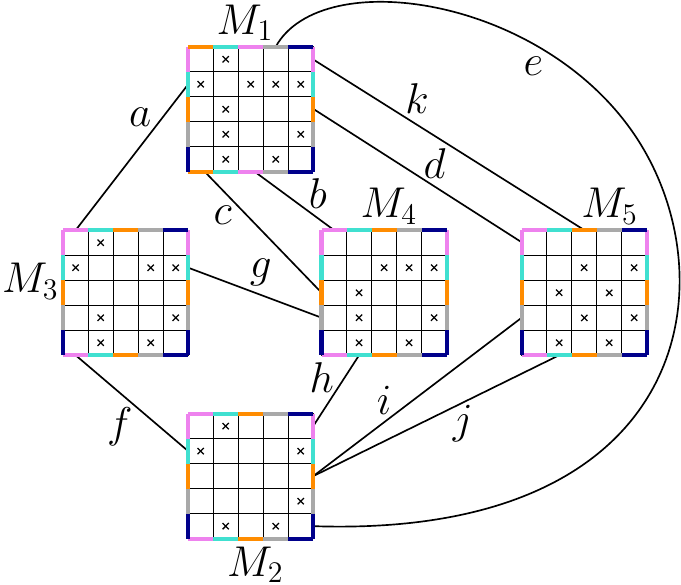}}
	\hfill
	\subfigure[\label{fi:matrix-simpq-b}{}]
	{\includegraphics[width=.38\textwidth,page=3]{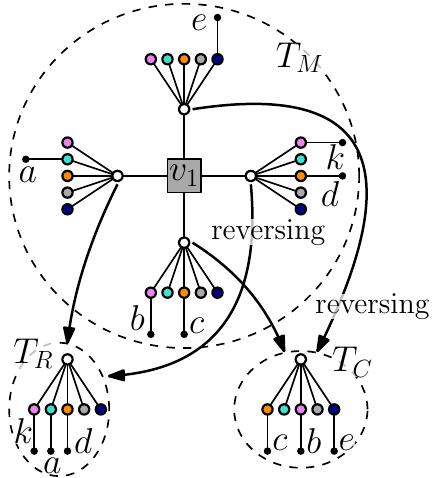}}
%	\hfil
%	\subfigure[\label{fi:matrix-simpq-c}{}]
%	{\includegraphics[width=.15\textwidth,page=4]{matrix-simPQ}}
	\caption{(a) An \rci planar graph $(G,S)$ that is not \nodetrix planar with fixed sides. (b) The constraint DAG $\mathcal{H}(v_1)$ associated with vertex $v_1$ of the equipped frame graph of $G$, corresponding to matrix $M_1$. P-nodes are circles, F-nodes are shaded boxes.}
	\label{fi:matrix-simpq}
\end{figure}

Each vertex $v$ of $G_F$ is associated with a \emph{constraint DAG} $\mathcal{H}(v)$ whose nodes are FPQ-trees. More precisely, the source vertex of $\mathcal{H}(v)$ is an FPQ-tree $T_M$ consisting of an F-node with four incident P-nodes; each of such P-nodes describes possible permutations for the rows or for the columns of the matrix $M_v$. Two P-nodes encode the permutations of the rows (on the left and right hand-side of $M_v$), and the other two P-nodes encode the permutations of the columns (on the top and bottom hand-side of $M_v$).
The source of $\mathcal{H}(v)$ has two adjacent vertices; one of these adjacent vertices is associated with an FPQ-tree $T_R$, and the other one is associated with an FPQ-tree $T_C$. $T_R$ specifies permutations for the rows of $M_v$, and $T_C$ specifies permutations for the columns of $M_v$, that must be respected by the P-nodes of the FPQ-tree in the root of $\mathcal{H}(v)$. We say that $T_R$ and $T_C$ define the \emph{coherence} between the permutations of the rows and the permutations of the columns, respectively.
Fig.~\ref{fi:matrix-simpq-a} shows a \nodetrix graph $(G,S)$ whose clusters have size $5$ and Fig.~\ref{fi:matrix-simpq-b} shows the constraint DAG $\mathcal{H}(v_1)$ associated with vertex $v_1$ of the equipped frame graph of $G$.
Note that $G$ is \rci planar but it is not \nodetrix planar with fixed sides: If we require the rows and the columns of $M_1$ to have the same permutation, it is easy to check that either a crossing between $b$ and $c$ or one between $d$ and $k$ occurs.
%Note that $G$ is not \nodetrix planar but it is \rci planar, indeed by permuting the order of the vertices of $M_1$ we can avoid the crossing between $b$ and $c$, but in this case a crossing between $d$ and $k$ would appear, unless we permute only the order of the columns of $M_1$.
Two arcs of Fig.~\ref{fi:matrix-simpq-b} are labeled \emph{reversing} because, for any given permutation of the rows (columns), the rows (columns) are encountered in opposite orders when walking around $M_1$.
Note that $\mathcal{H}(v)$ is an instance of \simfpqord.
%The following property is an immediate consequence of the definition of constraint DAG $\mathcal{H}(v)$.

\begin{property}\label{pr:constraintDAGlength1}
	For each vertex $v$ of $G_F$, $\mathcal{H}(v)$ is a $1$-fixed constraint.
\end{property}

Let $\mathcal{D}$ be the embedding DAG of $G_F$.
Each vertex $v$ of $G_F$ is associated with its constraint DAG $\mathcal{H}(v)$ and its embedding DAG $\mathcal{D}(v)$.
%An equipped frame graph $G_F=(V,E)$ is said to be \emph{simultaneous FPQ-planar} if it admits an embedding that simultaneously satisfies the FPQ-constraints given by the embedding DAG $\mathcal{D}(v)$ and the ones given by the constraint DAG $\mathcal{H}(v)$, for each $v\in V$.

\begin{restatable}{lemma}{lemmaInterplay}\label{le:interplay2}
%	\textsc{[*]}
	A biconnected \nodetrix graph with fixed sides is \rci planar if and only if its equipped frame graph is $1$-fixed constrained planar.
%	simultaneous FPQ-planar.
\end{restatable}
\begin{proof}
	Let $(G,S)$ be a biconnected \nodetrix graph with fixed sides, let $G_F=(V,E)$ be its equipped frame graph, and let $\mathcal{D}$ be the embedding DAG of $G_F$. For each $1 \le i \le |V|$, let $v_i$ be a vertex of $G_F$, let $\mathcal{H}(v_i)$ be its constraint DAG, and let $\mathcal{D}(v_i)$ be its embedding DAG.
	
	If $G_F$ is $1$-fixed constrained planar, $G_F$ admits an embedding that simultaneously satisfies the constraints given by the embedding DAG $\mathcal{D}(v_i)$ and the ones given by the constraint DAG $\mathcal{H}(v_i)$, for each $v_i\in V$. Since each $\mathcal{D}(v_i)$ is satisfied, the cyclic orders of the edges around the vertices of $G_F$ describe all the planar combinatorial embeddings of $G_F$. The constraint DAGs associated with each $v_i$ describe the constraints that allow to replace each vertex of $G_F$ with a matrix whose inter-cluster edges are incident to the sides as specified by the input. In particular, for each vertex $v_i\in G_F$ we replace its constraint DAG $\mathcal{H}(v_i)$ by a gadget $W^i$ that is built as follows (refer to Fig.~\ref{fi:constraintdag-gadget}). The F-node $\chi$ of the source tree $T_{M_i}$ is replaced with a wheel $H_\chi$ whose external cycle has four vertices $v_S$, where $S \in \{\texttt{T}, \texttt{R}, \texttt{B}, \texttt{L}\}$, one for each edge incident to $\chi$.
	Each P-node $\pi_S$ that is adjacent to $\chi$ in the source tree $T_{M_i}$ is represented in $W^i$ as a vertex $v_{\pi_S}$ that is connected to the corresponding $v_S$ ($S \in \{\texttt{T}, \texttt{R}, \texttt{B}, \texttt{L}\}$).
	Each node $\pi_S$ is adjacent to $k$ P-nodes $\rho_S^1, \dots, \rho_S^k$ in $T_{M_i}$. Each $\rho_S^j$ ($1 \le j \le k$, $S \in \{\texttt{T}, \texttt{R}, \texttt{B}, \texttt{L}\}$) is represented in $W^i$ as a vertex $v_S^j$. Finally, for each edge incident to $\rho_S^j$ in $T_{M_i}$ there is in $W^i$ an edge called \emph{spoke} incident to $v_S^j$.
	For example, Fig.~\ref{fi:constraintdag-gadget-b} shows the gadget corresponding to the constraint DAG of Fig.~\ref{fi:constraintdag-gadget-a}.
	
	Note that the trees $T_R$ and $T_C$ fix the permutations of the children of the two pairs of P-nodes $(\pi_\texttt{T},\pi_\texttt{B})$ and $(\pi_\texttt{R},\pi_\texttt{L})$ of $T_{M_i}$, that are guaranteed to be coherent because the constraints described by $\mathcal{H}(v_i)$ are satisfied. In $W^i$, the permutations of the pairs $(v_{\pi_\texttt{T}},v_{\pi_\texttt{B}})$ and $(v_{\pi_\texttt{R}},v_{\pi_\texttt{L}})$ are fixed consistently. Also, such permutations lead to a planar embedding because they satisfy the embedding DAG $\mathcal{D}(v_i)$.
	
	\begin{figure}[tb]
		\centering
		\subfigure[\label{fi:constraintdag-gadget-a}{}]
		{\includegraphics[width=.48\textwidth,page=1]{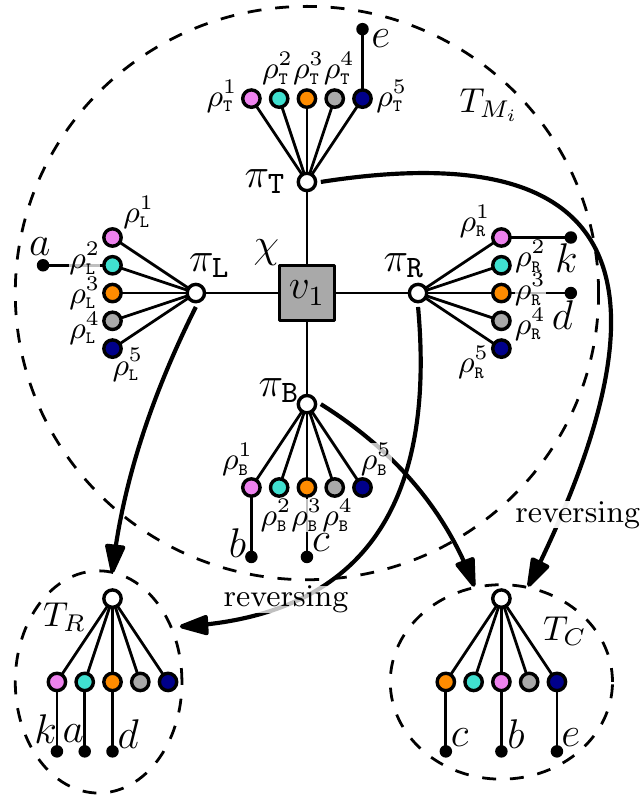}}
		\hfill
		\subfigure[\label{fi:constraintdag-gadget-b}{}]
		{\includegraphics[width=.48\textwidth,page=2]{constraintdag-gadget}}
		\caption{(a) The constraint DAG $\mathcal{H}(v_1)$ of Fig.~\ref{fi:matrix-simpq-b}. (b) The gadget $W^i$ replacing $\mathcal{H}(v_1)$.}
		\label{fi:constraintdag-gadget}
	\end{figure}
	
	By performing such a replacement for each vertex $v_i\in G_F$ and by connecting the spokes of the gadgets that correspond to the same edge, we obtain a planar graph $\hat{G_F}$.	
	In order to obtain a planar NodeTrix representation, we replace each gadget $W^i$ by a matrix as follows.
	
	Let $v_\texttt{T}$, $v_\texttt{R}$, $v_\texttt{B}$, $v_\texttt{L}$, be the vertices that are encountered by walking clockwise along the cycle of the wheel of $W^i$. Let $v_\texttt{T}^1, \dots, v_\texttt{T}^k$ be the clockwise order of the children of $v_{\pi_\texttt{T}}$, let $v_\texttt{R}^1, \dots, v_\texttt{R}^k$ be the clockwise order of the children of $v_{\pi_\texttt{R}}$, let $v_\texttt{B}^k, \dots, v_\texttt{B}^1$ be the clockwise order of the children of $v_{\pi_\texttt{B}}$, and let $v_\texttt{L}^k, \dots, v_\texttt{L}^1$ be the clockwise order of the children of $v_{\pi_\texttt{L}}$ in $W^i$. Replace $W^i$ by a matrix $M_i$ whose vertices are placed according to the permutations described for the columns and for the rows. The spokes of $W^i$ that are adjacent to $v_\texttt{T}^j$ ($1 \le j \le k$) are connected to $v_j$ on the top side of $M_i$. Analogously, the spokes of $W^i$ that are adjacent to $v_\texttt{R}^j$, $v_\texttt{B}^j$, or $v_\texttt{L}^j$, are connected to $v_j$ on the right, bottom, or left side of $M_i$, respectively.
	By performing this replacement for each gadget of $\hat{G_F}$, we obtain a planar NodeTrix representation $G$ of the $1$-fixed constrained planar graph $G_F$. It follows that, if $G_F$ is $1$-fixed constrained planar, $(G,S)$ is \rci planar.
	
	We now show that if $(G,S)$ is \rci planar, then $G_F$ is $1$-fixed constrained planar.
	Let $\Gamma$ be a planar \nodetrix representation of $G$ whose rows and columns permutations are independent. Replace each matrix $M_i$ of $\Gamma$ by a vertex $v_i$, and connect to it all the inter-cluster edges that are incident to $M_i$.
	We obtain a planar drawing $\Gamma'$ such that the cyclic order of the edges incident to each vertex $v_i$ of $\Gamma'$ reflects the cyclic order of the edges incident to matrix $M_i$ in $\Gamma$.
	Such an order satisfies the constraint DAG $\mathcal{H}(v_i)$ of $G_F$, and it also satisfies the embedding DAG $\mathcal{D}(v_i)$ because $\Gamma'$ is planar. Therefore, $G_F$ is $1$-fixed constrained planar.
\end{proof}

\noindent \textbf{{Testing RCI-NT Planarity:}} Based on Lemma~\ref{le:interplay2}, we shall test whether $(G,S)$ is \rci planar by testing whether $G_F$ is $1$-fixed constrained planar.

Observe that $\mathcal{H}(v)$ and $\mathcal{D}(v)$ have the same leaves, since they describe possible cyclic orders for the same set of inter-cluster edges, namely those incident to the matrix $M_v$ associated with $v$ in $G_F$, hence $\mathcal{H}(v)$ and $\mathcal{D}(v)$ are joinable instances of \simfpqord. Graph $G_F$ is $1$-fixed constrained planar if and only if it admits a planar embedding such that, for each vertex $v$ the cyclic order of the edges incident to $v$ satisfies both the constraints given by $\mathcal{H}(v)$ and the ones given by $\mathcal{D}(v)$. These constraints are described by the join DAG $\mathcal{J}(v)$ of $\mathcal{H}(v)$ and $\mathcal{D}(v)$ (i.e., $\mathcal{J}(v) = \mathcal{H}(v) \Join \mathcal{D}(v)$).
%
% For example, Fig.~\ref{fi:intersection-dag-a} shows the constraint DAG $\mathcal{H}(v_1)$ of vertex $v_1$ in the graph of Fig.~\ref{fi:example-a}, Fig.~\ref{fi:intersection-dag-b} shows the embedding DAG $\mathcal{D}(v_1)$ of $v_1$, and Fig.~\ref{fi:intersection-dag-c} illustrates the the join DAG $\mathcal{J}(v_1)$ of $\mathcal{H}(v_1)$ and $\mathcal{D}(v_1)$. Edges $b$ and $c$ are connected to the same P-node in $\mathcal{J}(v_1)$, because both in $\mathcal{H}(v_1)$ and in $\mathcal{D}(v_1)$ they are connected to a P-node and hence they can be arbitrarily permuted; while $a$ and $e$ are connected to an F-node in $\mathcal{J}(v_1)$, because their permutation is fixed according to $\mathcal{H}(v_1)$.
%Recall that for any vertex $v$ of $G_F$, $\mathcal{J}(v)$ is an instance of \simfpqord. Moreover, by definition of constraint DAG $\mathcal{H}(v)$ and embedding DAG $\mathcal{D}(v)$ of $v$, they do not have common sinks. The following is an immediate consequence of Properties~\ref{pr:constraintDAGlength1} and~\ref{pr:embeddingDAGlength1}, and of the definition of intersection DAG $\mathcal{I}(v)$.
%
The following property is implied by  Property~\ref{pr:embeddingDAGlength1}, Property~\ref{pr:constraintDAGlength1}, and Lemma~\ref{le:intersection-2fixed}.

\begin{property}\label{pr:intersectionDAGlength1}
	For each vertex $v$ of $G$, $\mathcal{J}(v)$ is $2$-fixed.
\end{property}

We can now exploit Theorem~\ref{th:constrained-embedding-quadratic}, and hence we can decide in $O(n_F^2)$ time whether $G_F$ is $1$-fixed constrained planar, where $n_F$ is the number of vertices of $G_F$. By Lemma~\ref{le:interplay2},
and since constructing $G_F$ may require $O(n^2)$ time,
the following theorem holds.
%By Property~\ref{pr:intersectionDAGlength1} we can exploit Theorem~\ref{th:constrained-embedding-quadratic}, thus we can decide in $O(n_F^2)$ time whether $G_F$ is $1$-fixed constrained planar, where $n_F$ is the number of vertices of $G_F$. By Lemma~\ref{le:interplay2}, the following theorem holds.

\begin{restatable}{theorem}{theoremRcintQuadratic}\label{th:rcint-quadratic}
	Let $(G,S)$ be a biconnected \nodetrix graph. \rcintplanarity can be tested in $O(n^2)$-time, where $n$ is the number of vertices of $G$.
\end{restatable}

%It may be worth noticing that the technique behind the proof of Theorem~\ref{th:rcint-quadratic} can be applied to test FPQ-constrained planarity for general graphs whose embedding constraints can be expressed by means of DAGs of FPQ-trees that satisfy Property~\ref{pr:intersectionDAGlength1}.

%It may be worth noticing that the technique behind the proof of Theorem~\ref{th:rcint-quadratic} can be applied to test FPQ-constrained planarity for general graphs whose embedding constraints can be expressed by means of DAGs of FPQ-trees that satisfy Property~\ref{pr:intersectionDAGlength1} and Lemma~\ref{le:twoOutgoing}.

%This observation can be regarded as a result of independent interest since it simplifies the application of Bläsius and Rutter~\cite{br-spoace-16} by skipping the so-called normalization phase for this family of graphs. See the appendix for more details.

\subsection{PolyLink Planarity Testing}\label{sse:other-hybrid}

An \rci planar graph has a planar \nodetrix representation where the inter-cluster edges are incident to different sides of a $4$-gon, and there are constraints that impose the vertices on opposite sides to respect the same permutation. We generalize this type of representation by replacing the $4$-gons with $\sigma$-gons having an even number of sides.
%For example, Figure~\ref{fi:introd-b} shows a planar \nt representation of the graph of Figure~\ref{fi:introd-a}, while Figure~\ref{fi:introd-c} shows a planar \polylink representation of the same graph.

%\begin{figure}[tb]
%	\centering
%	\subfigure[\label{fi:introd-a}{}]
%	{\includegraphics[width=.33\textwidth,page=11]{intersection-link-NodeTrix}}
%	\hfil
%	\subfigure[\label{fi:introd-b}{}]
%	{\includegraphics[width=.37\textwidth,page=12]{intersection-link-NodeTrix}}
%	\hfil 	
%	\subfigure[\label{fi:introd-c}{}]
%	{\includegraphics[width=.37\textwidth,page=15]{intersection-link-NodeTrix}}
%	\hfil 	
%	\subfigure[\label{fi:introd-d}{}]
%	{\includegraphics[width=.37\textwidth,page=13]{intersection-link-NodeTrix}}
%	\caption{(a) A flat clustered graph $(G,S)$. Clusters $S_1$ and $S_2$ are highlighted. (b) A \nt representation of $(G,S)$. (c) A \polylink representation of $(G,S)$. (d) An intersection-link representation of $(G,S)$.}
%	\label{fi:introd-simFPQ}
%\end{figure}

\begin{figure}[tb]
	\centering
	\subfigure[\label{fi:polylink-kp-1}{}]
	{\includegraphics[width=.47\textwidth,page=1]{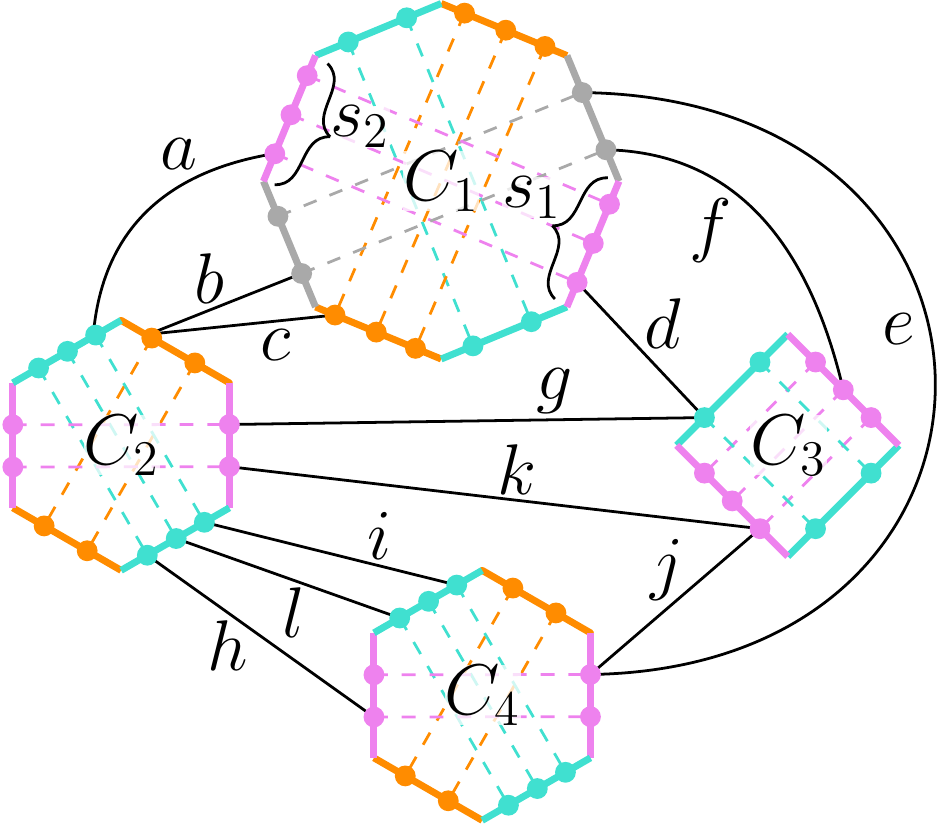}}
	\hfil
	\subfigure[\label{fi:polylink-kp-2}{}]
	{\includegraphics[width=.47\textwidth,page=3]{polylink-kp}}
	\caption{(a) A planar \polylink representation of a graph. Dashed segments connect two copies of a same vertex. (b) The constraint DAG $\mathcal{H}(v_1)$ associated with cluster $C_1$.}
	\label{fi:polylink-kp}
\end{figure}

A \emph{\polylink representation} of a flat clustered graph $(G,S)$ is such that each cluster $C$ in $S$ is represented as a polygon $P_C$ with an even number $\sigma_C$ of sides.
Each side of $P_C$ is associated with its antipodal side along the boundary of $P_C$. We group the set of vertices of $C$ into disjoint subsets; each subset is associated with at least one pair of opposite sides of $P_C$. Let $(s_1,s_2)$ be a pair of opposite sides of $P_C$ and let $u_1,\dots,u_s$ be the vertices of $G$ associated with $(s_1,s_2)$. A vertex $u_i$ ($1 \le i \le s$) is represented by a point on $s_1$ and a point on $s_2$; also, when walking clockwise around $P_C$ the vertices along $s_1$ are encountered in opposite order with respect to the vertices along $s_2$. For each inter-cluster edge $e=(u,v)$ such that $u$ is associated with $(s_1,s_2)$ and $v$ is associated with $(s_3,s_4)$, it is specified to which copy of $u$ (the one that lies on $s_1$ or the one that lies on $s_2$) and to which copy of $v$ (the one on $s_3$ or the one on $s_4$) edge $e$ must be incident.
%(Note that the representation of the intra-cluster edges is not relevant for the planarity testing problem.)

A \polylink representation is planar if no two inter-cluster edges cross. Figure~\ref{fi:polylink-kp-1} shows an example of a planar \polylink representation.

A flat clustered graph is a \emph{\polylink planar graph} if it admits a planar \polylink representation. We can test a flat clustered graph for \polylinkplanarity by generalizing the approach of Section~\ref{sse:rcint}. Namely, the constraint DAG $\mathcal{H}(v)$ associated with a cluster $C$ of a \polylink graph has a source $T_{P}$ that is an FPQ-tree, which consists of an F-node with $\sigma_C$ incident P-nodes, instead of four as in the case of \rcintplanarity. Each of such P-nodes describes the possible permutations of the vertices belonging to a side of the polygon $P_C$ representing $C$. For each pair $(s_1, s_2)$ of sides, the coherence between the order of the vertices belonging to $(s_1, s_2)$ is encoded by means of an FPQ-tree that is adjacent to the corresponding P-nodes of the source $T_{P}$. Figure~\ref{fi:polylink-kp-2} shows the constraint DAG $\mathcal{H}(v_1)$ associated with cluster $C_1$ of the graph in Figure~\ref{fi:polylink-kp-1}.
We say that a flat clustered graph having a \polylink representation is a biconnected \polylink graph if its equipped frame graph is biconnected.
%The embedding DAG $\mathcal{D}$ and the join DAG $\mathcal{J}$ of the equipped frame graph of a \polylink graph are defined analogously to the case of \rci planar graphs.
The same argument used to test \rcintplanarity leads to the following.

\begin{theorem}\label{th:polylink}
	Let $(G,S)$ be a biconnected \polylink graph. \polylinkplanarity can be tested in $O(n^2)$ time, where $n$ is the number of vertices of $G$.
\end{theorem}

Note that if the sides to which the inter-cluster edges must be incident are not specified, \polylinkplanarity is NP-complete. Indeed it becomes equivalent to \nodetrixplanarity with free sides if the polygons have maximum size four and each side is associated with the same set of vertices.

A flat clustered graph $(G,S)$  that admits an intersection-link representation is an \emph{intersection-link graph}.
%$(G,S)$ is \emph{intersection-link planar} if it admits a representation where no two inter-cluster edges cross.
Let $(G,S)$ be an intersection-link graph where $S$ is a partition of the vertices of $G$ and each cluster of $S$ is a clique. We recall that the clique planarity problem asks whether $(G,S)$ admits an intersection-link representation where no two inter-cluster edges cross. In~\cite{addfpr-ilrg-17} it is proved that if $(G,S)$ is clique-planar, then it admits a canonical intersection-link representation, i.e., an intersection-link representation where all vertices in a same cluster are isothetic unit squares whose upper-left corners are aligned along a line of slope one.

\begin{figure}[tb]
	\centering
	\subfigure[\label{fi:polylink-intersection-link-1}{}]
	{\includegraphics[width=.47\textwidth,page=2]{polylink-intersection-link}}
	\hfil
	\subfigure[\label{fi:polylink-intersection-link-2}{}]
	{\includegraphics[width=.47\textwidth,page=1]{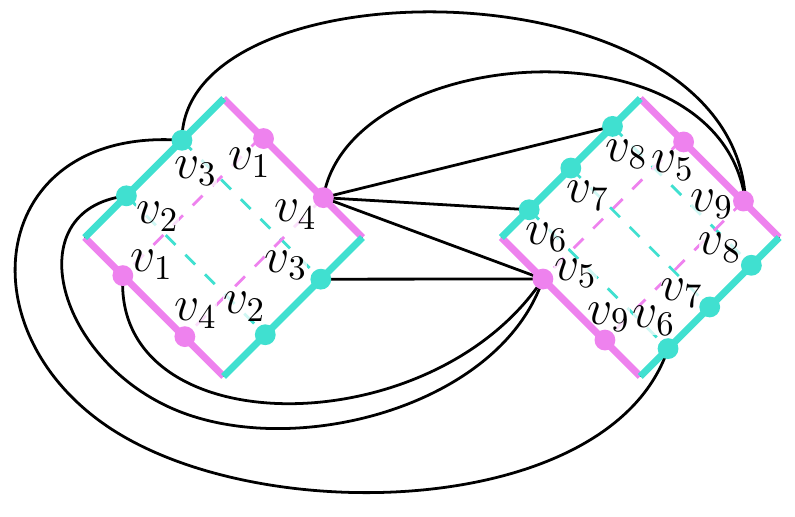}}
	\caption{(a) A canonical representation of an intersection-link graph $(G,S)$. (b) The corresponding \polylink representation.}
	\label{fi:polylink-intersection-link}
\end{figure}

Considering a canonical representation of an intersection-link graph, by walking along the boundary $B_C$ of a cluster $C$ of size $k$ (with $k\ge 3$), the cyclic order of its vertices is such that $k-2$ vertices appear twice and in opposite order, and these two sequences of vertices are separated by two single vertices that appear only once along $B_C$. We can hence model such an intersection-link graph as a \polylink graph where each cluster is a polygon with four sides: A pair of sides contains $k-2$ vertices, while the other two sides contain two vertices, each of which has incident inter-cluster edges only in one of the two sides.

An instance of \emph{clique planarity with fixed sides} specifies, for each inter-cluster edge $e=(u,v)$, the two sides of the unit squares representing $u$ and $v$ to which $e$ is incident.
See Figure~\ref{fi:polylink-intersection-link} for an example.

%A \polylink planar representation is a constrained version of a $(k,p)$-planar representation~\cite{dllrt-kpprhp-walcom19}.
%A $(k,p)$-planar representation of a flat clustered graph with clusters of size at most $k$ associates a convex region $R$ to each cluster $C$, and each vertex of $C$ is represented by at most $p$ points on the boundary of $R$.
%Each vertex that belongs to a cluster $C$ of a \polylink representation has pairs of copies on the boundary of its associated polygon $P_C$, but with the
%constraint that each pair of copies of a vertex $v$ must be antipodal on the boundary of $P_C$; conversely, in the $(k,p)$ representation model the different copies of the vertices can follow an arbitrary order along the boundary of a cluster.

An intersection-link graph is biconnected if its equipped frame graph is biconnected.
By exploiting the relation between \polylink graphs and intersection-link graphs, the following corollary holds.

%We conclude this section by observing that \polylink planar representations are also related to $(k,p)$ representations~\cite{dllrt-kpprhp-walcom19} and to intersection-link representations~\cite{addfpr-ilrg-17}. Details can be found in the appendix.

%\begin{corollary}\label{co:inters-polylink}
%	Let $(G,S)$ be a biconnected intersection-link graph and let $n_I=|S|$. There exists an $O(n_I^2)$-time algorithm that tests whether $(G,S)$ is intersection-link planar.
%\end{corollary}

\begin{corollary}\label{co:inters-polylink}
	Let $(G,S)$ be a biconnected intersection-link graph. Clique planarity with fixed sides can be tested in $O(n^2)$ time, where $n$ is the number of vertices of $G$.
\end{corollary}

We remark that clique planarity is NP-complete if the sides to which the inter-cluster edges are incident are not fixed~\cite{addfpr-ilrg-17}.

\section{Open Problems}\label{se:conclusions}

%This paper introduces a new definition of fixedness for \simfpqord, that can be used to strongly simplify the design of efficient algorithms for constrained planarity testing, including several variants of hybrid planarity testing. We envision that this simplification can be adopted also in other contexts.
The research in this paper suggests the following open problems:
	(i) Our hybrid planarity testing results assume the equipped frame graph to be biconnected. It would be interesting to study the connected case.
	(ii) What is the complexity of the \rowcolumn testing problem in the free sides scenario?
	(iii) Finally, it would be interesting to validate the \rci paradigm and the \polylink paradigm with user studies that compare them with the \nodetrix paradigm. Indeed, the class of  \rci planar graphs is a proper superclass of the planar \nodetrix graphs and if their readability is not significantly worse than in the standard \nodetrix model, the fact that they can be tested in polynomial time could be the starting point for developing new efficient visualization systems.

%\newpage
\bibliographystyle{splncs04}
\bibliography{biblio}

\end{document}